\documentclass[journal]{IEEEtran}
\usepackage{amsthm}
\usepackage{algorithm,algorithmic}
\newtheorem{thm}{Theorem}
\usepackage{bm}

\ifCLASSINFOpdf
\else
\fi
\begin{document}
\title{Approximate Message Passing with Unitary Transformation}
%
%On the Convergence of

\author{Qinghua~Guo~\IEEEmembership{}
        and~Jiangtao~Xi~\IEEEmembership{}
        %and~Jane~Doe,~\IEEEmembership{Life~Fellow,~IEEE}% <-this % stops a space
\thanks{The authors are with the School of Electrical, Computer and Telecommunications Engineering, University of Wollongong,
NSW, 2500, Australia, e-mail: \{qguo,jiangtao\}@uow.edu.au.}% <-this % stops a space
}


\markboth{}%
{Shell \MakeLowercase{\textit{et al.}}: Bare Demo of IEEEtran.cls for Journals}



\maketitle

\begin{abstract}
Approximate message passing (AMP) and its variants, developed based on loopy belief propagation, 
are attractive for estimating a vector $\textbf{x}$ from a noisy version of $\textbf{z}=\textbf{A}\textbf{x}$, 
which arises in many applications. For a large $\textbf{A}$ with i. i. d. elements, AMP can be characterized by the state evolution and exhibits fast convergence. However, it has been shown that, AMP may easily diverge for a generic $\textbf{A}$. In this work, we develop a new variant of AMP based on a unitary transformation of the original model (hence the variant is called UT-AMP), where the unitary matrix is available for any matrix $\textbf{A}$, e.g., the conjugate transpose of the left singular matrix of $\textbf{A}$, or a normalized DFT (discrete Fourier transform) matrix for any circulant $\textbf{A}$. We prove that, in the case of Gaussian priors, UT-AMP always converges for any matrix $\textbf{A}$. It is observed that UT-AMP is much more robust than the original AMP for `difficult' $\textbf{A}$ and exhibits fast convergence.

A special form of UT-AMP with a circulant $\textbf{A}$ was used in our previous work [\ref{ampeq}] for turbo equalization. This work extends it to a generic $\textbf{A}$,
and provides a theoretical investigation on the convergence.

\end{abstract}


% \fi
\IEEEpeerreviewmaketitle

\section{Introduction}
\IEEEPARstart{A}{proximate} message passing, developed based on loopy belief propagation, is an efficient approach to the estimation of a vector $\textbf{x}$ with independent elements $\{x_i\sim p(x_i)\}$ in the following model
\begin{equation}
\mathrm{\textbf{y}=\textbf{Ax}+\textbf{n}} \label{a1}
\end{equation}
where $\textbf{A}$ is a known matrix with size $M \times N$, the length of $\textbf{x}$ is $N$, $\textbf{y}$ denotes a length-$M$ observation vector and $\textbf{n}$ is a length-$M$ white Gaussian noise vector with zero mean and covariance matrix $\sigma^2\textbf{I}$ [\ref{ampa}]-[\ref{ampca}]. AMP was originally developed for compressive sensing based on model (\ref{a1}) [\ref{ampa}]-[\ref{ampc}], and then was extended to generalized AMP (GAMP) to accommodate more general distribution $p(y_i|(\textbf{A}\textbf{x})_i)$ which may not be Gaussian (where $y_i$ and $(\textbf{A}\textbf{x})_i$ denotes the $i$-th element in $\textbf{y}$ and $(\textbf{A}\textbf{x})$, respectively) [\ref{gampb}], [\ref{gampa}].
For a large $\textbf{A}$ with i.i.d. elements, AMP exhibits fast convergence which can be characterized by the state
evolution [\ref{ampa}], [\ref{gampa}]. However, for a generic $\textbf{A}$, the convergence of AMP cannot be
guaranteed. It has been shown that AMP may diverge for a benign matrix $\textbf{A}$ and it can
easily diverge for a `difficult' matrix A, e.g., non-zero mean, rank-deficient, column-correlated, or ill-conditioned
$\textbf{A}$ [\ref{swampa}], [\ref{dgamp}].

The fixed points and convergence of AMP were analyzed for an arbitrary matrix $\textbf{A}$ in [\ref{anamp}] and [\ref{anampc}] . Reference [\ref{anamp}] provides sufficient conditions for the convergence of AMP in the case of Gaussian priors $\{p(x_i)\}$. The convergence condition is closely related to the peak-to-average ratio of the squared singular values of a certain normalized $\textbf{A}$ for vector stepsize AMP algorithm, and is closely related to the peak-to-average ratio of the squared singular values of $\textbf{A}$ for scalar stepsize AMP algorithm. Damped AMP algorithms were proposed and the convergence can be guaranteed with sufficient damping, but the amount of damping 
grows with the peak-to-average ratio [\ref{anamp}]. Adaptive damping and mean removal mechanisms were introduced to (G)AMP in [\ref{dgamp}] to enhance the convergence speed. Compared to original AMP, swept AMP (SwAMP) in [\ref{swampa}], [\ref{swamp}] is much more robust to difficult $\textbf{A}$. However, SwAMP updates the relevant estimates sequentially (in contrast, AMP updates them in parallel), which restricts fast implementations. The global convergence of AMP with a generic $\textbf{A}$ for a generic prior $p(x_i)$ has not been understood [\ref{anamp}].

In this work, we present a new variant of AMP, which is developed based on the following unitary transformation 
of (\ref{a1})
\begin{equation}
\textbf{r}=\bm{\Lambda} \textbf{V} + \textbf{w} \label{wa}
\end{equation}
where $\textbf{r}=\textbf{U}^H\textbf{y}$, $\textbf{w}=\textbf{U}^H\textbf{n}$, and
\begin{equation}
\textbf{A}=\textbf{U}\bm{\Lambda}\textbf{V} \label{aw}
\end{equation}
with $\textbf{U}$ and $\textbf{V}$ being unitary matrices and $\bm{\Lambda}$ being a rectangular diagonal matrix. We note that, as $\textbf{U}^H$ is a unitary matrix, $\textbf{w}$ is still a zero mean Gaussian noise vector with the same covariance matrix as $\textbf{n}$ in (\ref{a1}). Eqn. (\ref{aw}) holds for any $\textbf{A}$ through the singular value decomposition (SVD). It is worth mentioning that, any circulant matrix $\textbf{A}$ ($M=N$) can be unitarily diagonalized by a discrete Fourier transform (DFT) matrix, so $\textbf{U}$ and $\textbf{V}$ can simply be the normalized DFT matrix and its inverse. In addition, $\textbf{r}$ and (the diagonal elements of) $\bm{\Lambda}$ can be calculated with the fast Fourier transform (FFT). A new variant of AMP is then developed based on model (\ref{wa}), which, for convenience, is called UT-AMP (where UT stands for unitary transformation) in this paper. 
It is interesting that, although unitary transformation does not change the singular values of $\textbf{\textbf{A}}$, 
we prove that UT-AMP converges for any $\textbf{A}$ in the case of Gaussian priors. Moreover, we show that the 
convergence speed of UT-AMP is related to a scalar $\alpha$ (see Theorem 1 and its proof). It is observed that UT-AMP is much more robust than the original AMP algorithms and exhibits fast convergence. It is noted that the SVD required for a non-circulant $\textbf{A}$ only needs to be carried out once, so UT-AMP is particularly suitable for applications with a fixed $\textbf{A}$ (e.g., turbo MIMO detection with/without quasi-static channels in communications). For applications with model (\ref{a1}) with a circulant $\textbf{A}$ (e.g., block transmission with cyclic prefix in communications), the unitary transformation can be efficiently performed with FFT, which makes UT-AMP very attractive, e.g., in equalization to combat intersymbol interference, as shown in our previous work [\ref{ampeq}].

Notations: Bold lowercase letters, e.g., $\textbf{c}$, are used to denote column vectors, and bold upper case letters, e.g., \textbf{C}, are used to denote matrices. The $i$-th element
in vector $\textbf{c}$ is denoted by $c_i$.  We use $\textbf{c}\cdot \textbf{d}$ and $\textbf{c}./ \textbf{d}$ to denote the elementwise product and division between two vectors $\textbf{c}$ and $\textbf{d}$, respectively. $\mathrm{|\textbf{C}|^2}$ represents the elementwise magnitude squared operation for matrix $\textbf{C}$. $\textbf{1}$, $\textbf{0}$ and $\textbf{I}$ represent an all-one column vector, an all-zero column vector, and an identity matrix with proper sizes depending on the context. The conjugate transpose is denoted by the superscript ``$^H$''.

\section {AMP with vector stepsizes and scalar stepsizes}

To facilitate comparisons with UT-AMP, we include the vector stepsize AMP (Algorithm 1) and the scalar stepsize AMP (Algorithm 2) [\ref{gampa}] in this section. In vector stepsize AMP, the function
$g_x(\textbf{q}, \bm{\tau}_q)$ returns a column vector whose $i$-th element, denoted by $[g_x(\textbf{q}, \bm{\tau}_q)]_i$, is given by     \begin{equation}
[g_x(\textbf{q}, \bm{\tau}_q)]_i=\frac{\int x_i p(x_i)\mathcal{N}(x_i;q_i,\tau_{q_i})dx_i}{\int p(x_i)\mathcal{N}(x_i;q_i,\tau_{q_i})dx_i}  \label{ac}
\end{equation}
where $\mathcal{N}(x_i;q_i,\tau_{q_i})$ denotes a Gaussian distribution with $x_i$ as random variable, $q_i$ as mean, and $\tau_{q_i}$ as variance. Eqn. (\ref{ac}) can be interpreted as
the MMSE (minimum mean square error) estimation of $x_i$ based on the following model
\begin{equation}
q_i=x_i + \varpi \label{ad}
\end{equation}
where $x_i \sim p(x_i)$ and $\varpi$ is a Gaussian noise with mean zero and variance $\tau_{q_i}$.The function $g'_x(\textbf{q}, \bm{\tau}_q)$ returns a column vector, and the $i$-th element is denoted by $[g'_x(\textbf{q}, \bm{\tau}_q)]_i$ where the derivative is with respect to $q_i$. It is not hard to show that $\tau_{q_i}[g'_x(\textbf{q}, \bm{\tau}_q)]_i$ is the
a posteriori variance of $x_i$ with model (\ref{ad}). Note that $g_x(\textbf{q}, \bm{\tau}_q)$ can also be changed for MAP (maximum
a posteriori) estimation of $\textbf{x}$ [\ref{gampa}].

Scalar stepsize AMP can be obtained from vector stepsize AMP by forcing the elements of each variance vector to be 
the same, so that the multiplications of a matrix with a vector in updating $\bm{\tau}_p$ and $\bm{\tau}_q$ 
are avoided (compare Lines 1 and 5 in both algorithms). The function $g_x(\textbf{q}, {\tau}_q)$ is the same as $g_x(\textbf{q}, \bm{\tau}_q)$ in vector stepsize AMP except that all the Gaussian distributions $\{\mathcal{N}(x_i; q_i, \tau_q)\}$ share the same variance $\tau_q$.

\begin{algorithm}
\caption{Vector Stepsize AMP}
Initialize $\bm{\tau}^{(0)}_x$ (with elements larger than 0) and $\textbf{x}^{(0)}$. Set $\textbf{s}^{(-1)}=\textbf{0}$ and $t=0$.

Repeat

1. ${\bm{\tau}_p=|\textbf{A}|^2\bm{\tau}^t_x}$

2. ${\textbf{p}=\textbf{Ax}^t-\bm{\tau}_p \cdot \textbf{s}^{t-1}}$

3. ${\bm{\tau}_s=\textbf{1}./(\bm{\tau}_p+\sigma^2 \textbf{1})}$

4. ${\textbf{s}^{t}=\bm{\tau}_s \cdot (\textbf{y}-\textbf{p})}$

5. $\textbf{1}./\bm{\tau}_q=|\textbf{A}^H|^2\bm{\tau}_s$

6. ${\textbf{q}=\textbf{x}^t+\bm{\tau}_q \cdot \textbf{A}^H\textbf{s}^t}$

7. ${\bm{\tau}_x^{t+1}=\bm{\tau}_q \cdot g'_x(\textbf{q},\bm{\tau}_q)}$

8. ${\textbf{x}^{t+1}=g_x(\textbf{q},\bm{\tau}_q)}$

9. $t=t+1$

Until terminated

\end{algorithm}

\begin{algorithm}
\caption{Scalar Stepsize AMP}
Initialize ${\tau}^{(0)}_x>0$ and $\textbf{x}^{(0)}$. Set $\textbf{s}^{(-1)}=\textbf{0}$ and $t=0$.

Repeat

1. ${{\tau}_p=(1/M)|\textbf{A}|_F^2{\tau}^t_x}$

2. ${\textbf{p}=\textbf{Ax}^t-\tau_p\textbf{s}^{t-1}}$

3. ${\tau}_s=1/({\tau}_p+\sigma^2 )$

4. ${\textbf{s}^{t}={\tau}_s (\textbf{y}-\textbf{p})}$

5. ${1/{\tau}_q=(1/N)|\textbf{A}^H|_F^2{\tau_s}}$

6. ${\textbf{q}=\textbf{x}^t+{\tau}_q\textbf{A}^H\textbf{s}^t}$

7. ${\tau}_x^{t+1}=({\tau}_q/N)\textbf{1}^H g'_x(\textbf{q},{\tau}_q)$

8. ${\textbf{x}^{t+1}=g_x(\textbf{q},{\tau}_q)}$

9. $t=t+1$

Until terminated

\end{algorithm}

\begin{algorithm}
\caption{UT-AMP}

Unitary transform: ${\textbf{r}=\textbf{U}^H\textbf{y}=\bm{\Lambda} \textbf{V}\textbf{x}+\textbf{w}}$, where $\mathrm{\textbf{A}}=\mathrm{\textbf{U}\bm{\Lambda} \textbf{V}}$.

Define vectors $\bm{\lambda}_p = \bm{\Lambda} \bm{\Lambda}^H \textbf{1}$ and $\bm{\lambda}_s = \bm{\Lambda}^H\bm{\Lambda} \textbf{1}$.

Initialize ${\tau}^{(0)}_x>0$ and $\textbf{x}^{(0)}$. Set $\textbf{s}^{(-1)}=\textbf{0}$ and $t=0$.

Repeat

1. ${\bm{\tau}_p={\tau}^t_x}\bm{\lambda}_p$

2. ${\textbf{p}=\bm{\Lambda} \textbf{V}\textbf{x}^t-\bm{\tau}_p \cdot \textbf{s}^{t-1}}$

3. ${\bm{\tau}_s=\textbf{1}./(\bm{\tau}_p+\sigma^2 \textbf{1})}$

4. ${\textbf{s}^{t}=\bm{\tau}_s \cdot (\textbf{r}-\textbf{p})}$

5. ${1}/{\tau}_q=(1/N)\bm{\lambda}_s^H \bm{\tau}_s$

6. $\textbf{q}=\textbf{x}^t+ {\tau}_q (\textbf{V}^H\bm{\Lambda}^H\textbf{s}^t)$

7. ${\tau}_x^{t+1}=({\tau}_q /N) \textbf{1}^H g'_x(\textbf{q},{\tau}_q)$

8. ${\textbf{x}^{t+1}=g_x(\textbf{q},{\tau}_q)}$

9. $t=t+1$

Until terminated

\end{algorithm}

\section{UT-AMP and Its Convergence}

\subsection{Derivation of UT-AMP}

As any matrix $\textbf{A}$ can have the decomposition $\textbf{A}=\textbf{U} \bm{\Lambda} \textbf{V}$, we first perform a unitary transformation with $\textbf{U}^H$ to (\ref{a1}), yielding
\begin{equation}
{\textbf{r}=\textbf{U}^H\textbf{y}=(\bm{\Lambda} \textbf{V})\textbf{x}+\textbf{w}} \label{newa} \end{equation}
where $\bm{\Lambda}$ is an $M \times N$ rectangular diagonal matrix. Then the vector stepsize AMP can be applied to (\ref{newa}) where the system matrix becomes a special matrix $\bm{\Lambda} \textbf{V}$. Note that \begin{equation}
|\textbf{C}|^2 \textbf{d}=(\textbf{C}~Diag(\textbf{d})~\textbf{C}^H )_{D}\textbf{1} \label{ea} \end{equation}
where $Diag(\textbf{d})$ returns a diagonal matrix with the elements of $\textbf{d}$ on its diagonal, and $(\textbf{B})_D$ returns a diagonal matrix by forcing the off-diagonal elements of $\textbf{B}$ to zero.
Now suppose we have a variance vector $\bm{\tau}_x^t$. According to Line 1 in vector stepsize AMP and using (\ref{ea}), we have
\begin{equation}
\bm{\tau}_p = (\bm{\Lambda}\textbf{V}~Diag(\bm{\tau}_x^t)~\textbf{V}^H \bm{\Lambda}^H)_{D}\textbf{1}. \label{eb} \end{equation}
In attempting to reduce the computational complexity, we can find that if $\bm{\tau_x}^t$ has a form of $\gamma \textbf{1}$, the calculation of (\ref{eb}) can be significantly reduced.
This motivates the replacement of $\bm{\tau_x}^t$ with $\tau_x^t \textbf{1}$  where $\tau_x^t$ is the average of the elements of $\bm{\tau_x}^t$. So (\ref{eb}) is reduced to
\begin{equation}
\bm{\tau}_p = \tau_x^t \bm{\Lambda}\bm{\Lambda}^H\textbf{1} \label{eb1} \end{equation}
which is Line 1 in UT-AMP. Lines 2, 3 and 4 in UT-AMP can be obtained according to Lines 2, 3, 4 in vector stepsize AMP by simply replacing $\textbf{A}$ with $\bm{\Lambda} \textbf{V}$.
According to (\ref{ea}) again, Line 5 in vector stepsize AMP with matrix $\bm{\Lambda} \textbf{V}$ can be represented as
\begin{equation}
\textbf{1}./\bm{\tau}_q = (\textbf{V}^H \bm{\Lambda}^H~Diag(\bm{\tau_p})~\bm{\Lambda}\textbf{V})_{D}\textbf{1}.  \label{eg}
\end{equation}
In order to reduce the computational complexity, we can replace the diagonal matrix $\bm{\Lambda}^H~Diag(\bm{\tau_p})~\bm{\Lambda}$ in (\ref{eg}) with a scaled identity matrix $\rho \textbf{I}$ where
$\rho$ is the average of the diagonal elements of $\bm{\Lambda}^H~Diag(\bm{\tau_p})~\bm{\Lambda}$, i.e.,
\begin{equation}
\rho= (1/N) \textbf{1}^H \bm{\Lambda}^H \bm{\Lambda} \bm{\tau_p}.
\end{equation}
Hence (\ref{eg}) is reduced to Line 5 in UT-AMP.
Line 6 can be obtained from Line 6 in vector stepsize AMP by replacing $\textbf{A}$ with $\bm{\Lambda} \textbf{V}$.
Compared with Line 7 in vector stepsize AMP, an additional average operation is performed in Line 7 in UT-AMP to meet the requirement of a scalar $\tau^t_x$ in Line 1.
We note that the average operation is not necessarily in Line 7 as we can also put the additional average operation in Line 1. Line 8 in UT-AMP is the same as Line 8 in vector
stepsize AMP except that $\tau_q$ is a scalar.

\textbf{Remarks}:
\begin{itemize}
  \item One may try to get another variant of AMP by applying the scalar stepsize AMP to model (\ref{newa}), i.e., replacing $\textbf{A}$ with $\textbf{U}^H\textbf{A}$ and replacing $\textbf{y}$ with $\textbf{r}=\textbf{U}^H\textbf{y}$ in scalar stepsize AMP. It is interesting that the obtained algorithm will remain exactly the same as the original scalar stepsize AMP as $\textbf{U}^H$ will be canceled out 
  in scalar stepsize AMP. This means that unitary transformation have no impact on the convergence of scalar stepsize AMP.
  \item By the name, in vector stepsize AMP, $\bm{\tau}_x^t$, $\bm{\tau}_p$, $\bm{\tau}_s$, and $\bm{\tau}_q$ are all vectors, and in scalar step size AMP, the corresponding $\tau_x^t$, $\tau_p$, $\tau_s$, and $\tau_q$ are all scalars. In contrast, UT-AMP has two scalars ${\tau_x^t}$ and $\tau_q$ and two vectors $\bm{\tau}_p$ and $\bm{\tau}_s$.
  \item If $\textbf{A}$ is a circulant matrix, UT-AMP is very attractive as $\textbf{U}$ and $\textbf{V}$ can be simply a DFT matrix 
  and its inverse, and the diagonal elements of $\bm{\Lambda}$ can be calculated with FFT. Moreover, the multiplications of matrix and vector in UT-AMP can be implemented with FFT as well, leading to very low complexity.
  \item If $\textbf{A}$ is non-circulant and its SVD required in UT-AMP is available, the complexity per iteration of the UT-AMP is lower than that of vector stepsize AMP as the multiplications of matrix with vector are avoided in Lines 1 and 5. The complexity of UT-AMP is slightly higher than that of the scalar stepsize AMP due to the vector operations in Lines 1 and 5. Hence, UT-AMP is particularly suitable for applications with fixed $\textbf{A}$ as SVD only needs to be carried out once.
  \item Most importantly, it is observed that UT-AMP is robust to `difficult' matrix $\textbf{A}$ and exhibits fast convergence.
\end{itemize}

\subsection {Convergence of UT-AMP}

\begin{thm}
UT-AMP converges for any $\textbf{A}$ in the case of Gaussian priors.
\end{thm}
\begin{proof}
See Appendix A.
\end{proof}

It can be seen from the proof of Theorem 1 that, the convergence speed of UT-AMP is related to a parameter $\alpha$ given in (\ref{appexr}) in Appendix A.

Similar to the original AMP, the convergence of UT-AMP for a generic prior is unknown, which remains as future work. 
It is also interesting to investigate the convergence of swept UT-AMP. 

It is observed that UT-AMP is robust to `difficult' $\textbf{A}$  e.g., non-zero mean, rank-deficient, column-correlated, or ill-conditioned $\textbf{A}$, under which the original AMP often diverges. The special form of UT-AMP with a circulant $\textbf{A}$ in the case of discrete priori distributions has been used in [\ref{ampeq}] for equalization, where the channel matrix $\textbf{A}$ is ill conditioned, and the use of original AMP will diverge.  Various numerical examples will be provided in the full version of this paper.

\section{Conclusion}
In this work, we have developed a new AMP variant UT-AMP for a generic matrix $\textbf{A}$. It has been shown that UT-AMP always converges in the case of Gaussian priors for any $\textbf{A}$. It is observed that UT-AMP is robust to difficult $\textbf{A}$ and exhibits fast convergence.

% or


\appendices
\section{Proof of Theorem 1}

We assume that
\begin{equation}
p(\textbf{x}) \sim \mathcal{N} (\textbf{x};\textbf{x}^0, Diag(\bm{\tau}_x^0))
\end{equation}
where $\textbf{x}^0$ and $\bm{\tau}_x^0$ are the a priori mean vector and variance vector for $\textbf{x}$, respectively.

%\begin {equation}
%{\tau}_x^{t+1}= (1/N)\textbf{1}^H (\textbf{1}./(\textbf{1}./\bm{\tau}_x^a + \textbf{1}./\bm{\tau}_q))
%\end{equation}

Similar to the proof in [\ref{anamp}], it can be proven that the variance $\tau_x^t$ of UT-AMP for any $\textbf{A}$ always converges to a fixed point denoted by $\tau_x$. Next, we prove the convergence of $\textbf{x}_t$.

Define a diagonal matrix
\begin{equation}
\textbf{D}=Diag(\bm{\tau}_s)=(\tau_x^t \bm{\Lambda} \bm{\Lambda}^H+\sigma^2 \textbf{I})^{-1}.
\end{equation}
Then, according to the UT-AMP algorithm, we have
\begin{equation}
\textbf{s}^t=\tau_x^t\textbf{D} \bm{\Lambda}\bm{\Lambda}^H\textbf{s}^{t-1}-\textbf{D}\bm{\Lambda}\textbf{V}\textbf{x}^t +\textbf{D}\textbf{r},
\end{equation}
\begin{eqnarray}
1/ \tau_q &=& (1/N)\textbf{1}^H \bm{\Lambda}^H \textbf{D} \bm{\Lambda} \textbf{1} \nonumber     \\
&=&\frac{1}{N}\sum_{i=1}^{min\{M,N\}} \frac{|\lambda_i|^2}{\tau_x^t|\lambda_i|^2+\sigma^2},
\end{eqnarray}
with $\lambda_i$ being the $(i, i)$-th elements of $\bm{\Lambda}$.
and
\begin{eqnarray}
\textbf{x}^{t+1}&=&\tau_x^{t+1}(\textbf{q}/\tau_q + \textbf{x}^0./ \bm{\tau}_x^0) \nonumber   \\
&=&(\tau_x^{t+1}/\tau_q) \textbf{x}^t + \tau_x^{t+1}\textbf{V}^H\bm{\Lambda}^H\textbf{s}^t + \tau_x^{t+1} \textbf{x}^0./ \bm{\tau}_x^0 \nonumber \\
&=&\tau_x^{t+1}\tau_x^t \textbf{V}^H\bm{\Lambda}^H \textbf{D}\bm{\Lambda}\bm{\Lambda}^H \textbf{s}^{t-1}  \nonumber \\
&&~~~~~+ (\alpha \textbf{I}-\tau_x^{t+1}\textbf{V}^H\bm{\Lambda}^H \textbf{D}\bm{\Lambda}\textbf{V})\textbf{x}^t + \textbf{b},
\end{eqnarray}
where $\textbf{b}$ is an appropriate vector. Define
\begin{equation}
\widetilde{\alpha}=\frac {\tau_x^{t+1}}{\tau_q} = \frac{1}{N}\sum_{i=1}^{min\{M,N\}} \frac{\tau_x^{t+1}|\lambda_i|^2}{\tau_x^t|\lambda_i|^2+\sigma^2}. \label{alpha}
\end{equation}
The iteration of $\textbf{x}^t$ and $\textbf{s}^t$ in UT-AMP can be described as
\begin{equation}
\left[
  \begin{array}{c}
    \textbf{s}^t \\
    \textbf{x}^{t+1} \\
  \end{array}
\right] =
\underbrace{\left[
  \begin{array}{cc}
    \textbf{C}_a & \textbf{C}_b \\
    \textbf{C}_c & \textbf{C}_d \\
  \end{array}
\right]}_\textbf{C}
\left[
  \begin{array}{c}
    \textbf{s}^{t-1} \\
    \textbf{x}^{t} \\
  \end{array}
\right] + \textbf{e}
\end{equation}
where $\textbf{e}$ is an appropriate vector. Matrix $\textbf{C}$ has two diagonal sub-matrices
\begin{equation}
\textbf{C}_a=\tau_x^t\textbf{D} \bm{\Lambda}\bm{\Lambda}^H,
\end{equation}
and
\begin{equation}
\textbf{C}_b=-\textbf{D}\bm{\Lambda}\textbf{V},
\end{equation}
and the other two sub-matrices can be represented as
\begin{equation}
\textbf{C}_c=\tau_x^{t+1}\tau_x^t \textbf{V}^H\bm{\Lambda}^H \textbf{D}\bm{\Lambda}\bm{\Lambda}^H,
\end{equation}
and
\begin{equation}
\textbf{C}_d=\widetilde{\alpha} \textbf{I}-\tau_x^{t+1}\textbf{V}^H\bm{\Lambda}^H \textbf{D}\bm{\Lambda}\textbf{V}.
\end{equation}

Next, we find the eigenvalues of matrix $\textbf{C}$, i.e., the roots of the following polynomial
\begin{equation}
h(\eta)=|\eta \textbf{I} - \textbf{C}|=\left|\begin{array}{cc}
                                  \eta \textbf{I} -\textbf{C}_a & \textbf{C}_b \\
                                  \textbf{C}_c& \eta \textbf{I} - \textbf{C}_d
                                \end{array}
\right| = 0. \label{appexa}
\end{equation}
We note that the identity matrices in (\ref{appexa}) have different sizes (i.e., the use of $\textbf{I}$ is abused for notation
simplification).
As $\textbf{C}_a$ is a diagonal matrix (with non-negative elements), a diagonal matrix $\omega \textbf{I}$ can be used to guarantee that $\eta\textbf{I}-\textbf{C}_a + \omega \textbf{I}$ is invertible. Define a new polynomial
\begin{equation}
h^a(\eta)=\left|\begin{array}{cc}
                                  \eta \textbf{I} -\textbf{C}_a +\omega \textbf{I} & \textbf{C}_b \\
                                  \textbf{C}_c& \eta \textbf{I} - \textbf{C}_d
                                \end{array}
\right|.
\end{equation}
Clearly the roots of $h^a(\eta)$ with $\omega =0$ are the eigenvalues of matrix $\textbf{C}$. It can be shown that $h^a(\eta)$ can be rewritten as
\begin{eqnarray}
h^a(\eta) \!\!\!\!\!\!&=&\!\!\!\!\!\!| \eta \textbf{I} -\textbf{C}_a +\omega \textbf{I}|\times|\eta \textbf{I} - \textbf{C}_d - \textbf{C}_c (\eta \textbf{I} -\textbf{C}_a +\omega \textbf{I})^{-1} \textbf{C}_b | \nonumber \\
\!\!\!\!\!\!&=&\!\!\!\!\!\!| \eta \textbf{I} -\tau_x^t\textbf{D} \bm{\Lambda}\bm{\Lambda}^H +\omega \textbf{I}| \times \nonumber \\
&&\!\!\!\!\!\!|\textbf{V}^H|\times|(\eta-\widetilde{\alpha}) \textbf{I} + \tau_x^{t+1}\bm{\Lambda}^H \textbf{D}\bm{\Lambda} + \tau_x^{t+1}\tau_x^t \bm{\Lambda}^H \textbf{D}\bm{\Lambda}\bm{\Lambda}^H \nonumber \\
&&(\eta \textbf{I} -\tau_x^t\textbf{D} \bm{\Lambda}\bm{\Lambda}^H +\omega \textbf{I})^{-1} \textbf{D}\bm{\Lambda}|\times|\textbf{V}|. \label{appexa1}
\end{eqnarray}
As $\textbf{V}$ is a unitary matrix, $|\textbf{V}^H|=|\textbf{V}|=1$. So they can be removed from (\ref{appexa1}). Note that $\bm{\Lambda}$ is a rectangular diagonal matrix with size $M \times N$ and all the matrices left in (\ref{appexa1}) are diagonal.

\textbf{Case 1}: $M=N$. In this case, $\bm{\Lambda}$ is a diagonal matrix. Define vector $\bm{\beta}=[\beta_1, ..., \beta_N]^T$ whose elements are the diagonal elements of $\tau_x^t\textbf{D} \bm{\Lambda}\bm{\Lambda}^H$, i.e.,
\begin{equation}
\beta_i=\frac{\tau_x^t|\lambda_i|^2}{\tau_x^t|\lambda_i|^2+\sigma^2}.
\end{equation}
where we can see that $0\leq\beta_i<1$. Hence (\ref{appexa1}) can be rewritten as
\begin{eqnarray}
h^a(\eta)\!\!\!\!\!\!&=&\!\!\!\!\!\!\prod_{i=1}^N (\eta-\beta_i+w)\big((\eta-\widetilde{\alpha}) + \frac{\tau_x^{t+1}}{\tau_x^{t}}(\beta_i +
\frac{\beta_i^2}{\eta-\beta_i+w}) \big) \nonumber \\
\!\!\!\!\!\!&=&\!\!\!\!\!\!\prod_{i=1}^N (\eta-\beta_i+w)\big((\eta-\widetilde{\alpha}) + \frac{\tau_x^{t+1}}{\tau_x^{t}}\beta_i\big)+\frac{\tau_x^{t+1}}{\tau_x^{t}}
{\beta_i^2}.
\end{eqnarray}
Letting $\omega=0$, we have
\begin{equation}
h(\eta)=\prod_{i=1}^N (\eta-\beta_i)\big((\eta-\widetilde{\alpha}) + \frac{\tau_x^{t+1}}{\tau_x^{t}}\beta_i\big)+\frac{\tau_x^{t+1}}{\tau_x^{t}}
{\beta_i^2}. \label{appexah}
\end{equation}
As the variance always converges, we have $\tau_x^{t+1}=\tau_x^{t}=\tau_x$ after a
certain number of iterations. Then (\ref{appexah}) can be reduced to
\begin{equation}
h(\eta)=\prod_{i=1}^N (\eta^2-{\alpha}\eta+\alpha\beta_i)
\end{equation}
where from (\ref{alpha})
\begin{equation}
\alpha= \frac{1}{N}\sum_{i=1}^{min\{M,N\}} \frac{\tau_x|\lambda_i|^2}{\tau_x|\lambda_i|^2+\sigma^2}. \label{appexr}
\end{equation}
which indicates that $0<\alpha<1$. Hence the eigenvalues are given by
\begin{equation}
\eta_{i(1,2)}=\frac{-{\alpha}\pm \sqrt{\alpha^2-4\alpha\beta_i}}{2}
\end{equation}
for $i=1, 2, ..., N.$ Recall that $0\leq\beta_i<1$. If $\alpha\geq\beta_i/4$, the eigenvalues are real, and it can be easily shown that
\begin{equation}
|\eta_i|\leq \alpha<1,
\end{equation}
where the equality holds when $\beta_i=0$.
If $\alpha<\beta_i/4$, the eigenvalues are complex valued, and it can be shown that
\begin{equation}
|\eta_i|\leq \alpha\beta_i<1.
\end{equation}

\textbf{Case 2}: $M>N$. In this case, $\bm{\Lambda}$ is a `tall' rectangular diagonal matrix. We define diagonal matrix
$\widetilde{{\bm{\Lambda}}}$ with size $N \times N$ as the upper part of $\bm{\Lambda}$, and define diagonal matrix $\widetilde{\textbf{D}}$ with size $N \times N$ as the upper left part of $\textbf{D}$ (whose size is $M \times M$). It is not hard to show that
\begin{eqnarray}
h^a(\eta)
\!\!\!&=&\!\!\!| \eta \textbf{I}_{M-N} +\omega \textbf{I}_{M-N}|\times| \eta \textbf{I} -\tau_x^t\widetilde{\textbf{D}} \widetilde{\bm{\Lambda}}\widetilde{\bm{\Lambda}}^H +\omega \textbf{I}| \nonumber \\
&&\times|(\eta-\alpha) \textbf{I} + \tau_x^{t+1}\widetilde{\bm{\Lambda}}^H \widetilde{\textbf{D}}\widetilde{\bm{\Lambda}} + \tau_x^{t+1}\tau_x^t \widetilde{\bm{\Lambda}}^H \widetilde{\textbf{D}}\widetilde{\bm{\Lambda}}\widetilde{\bm{\Lambda}}^H \nonumber \\
&&(\eta \textbf{I} -\tau_x^t\widetilde{\textbf{D}} \widetilde{\bm{\Lambda}}\widetilde{\bm{\Lambda}}^H +\omega \textbf{I})^{-1} \widetilde{\textbf{D}}\widetilde{\bm{\Lambda}}|.
\end{eqnarray}
After some manipulations, we have
\begin{equation}
h(\eta)=\eta^{M-N}\prod_{i=1}^N (\eta^2-\alpha\eta+\alpha\beta_i). \label{appexcb}
\end{equation}
Hence the eigenvalues are the same as those in Case 1 except that $M-N$ eigenvalues are zero.

\textbf{Case 3}: $M<N$. In this case, $\bm{\Lambda}$ is a `fat' rectangular diagonal matrix. Define diagonal matrix $\overline{\bm{\Lambda}}$ with size$M \times M$ as the left part of $\bm{\Lambda}$.  We can show that
\begin{eqnarray}
h^a(\eta)
\!\!\!&=&\!\!\!| \eta \textbf{I} -\tau_x^t{\textbf{D}} \overline{\bm{\Lambda}}\overline{\bm{\Lambda}}^H +\omega \textbf{I}| \times|(\eta-\alpha) \textbf{I} +
\tau_x^{t+1}\overline{\bm{\Lambda}}^H {\textbf{D}}\overline{\bm{\Lambda}} + \nonumber \\
&& \tau_x^{t+1}\tau_x^t \overline{\bm{\Lambda}}^H {\textbf{D}}\overline{\bm{\Lambda}}\overline{\bm{\Lambda}}^H (\eta \textbf{I} -\tau_x^t{\textbf{D}}
\overline{\bm{\Lambda}}\overline{\bm{\Lambda}}^H +\omega \textbf{I})^{-1} {\textbf{D}}\overline{\bm{\Lambda}}| \nonumber \\
&& \times | (\eta-\alpha) \textbf{I}_{N-M}|.
\end{eqnarray}
Then, we can have
\begin{equation}
h(\eta)=(\eta-\alpha)^{N-M}\prod_{i=1}^M (\eta^2-\alpha\eta+\alpha\beta_i).
\end{equation}
The eigenvalues are the same as those in Case 1 except that $N-M$ eigenvalues are $\alpha$.

The above shows that $|\eta_i| \leq \alpha$ for all the cases (noting that $\alpha\beta_i<\alpha$) and any matrix $\textbf{A}$. Because $\alpha$ is smaller than 1, the algorithm converges for any $\textbf{A}$.

%\section{}

%% use section* for acknowledgment
%\section*{Acknowledgment}


\ifCLASSOPTIONcaptionsoff
  \newpage
\fi

\end{document}